\documentclass[conference]{IEEEtran}
\usepackage{amsmath, array}
\usepackage{amsfonts}
\usepackage{fancyhdr}
\usepackage[cyr]{aeguill}
\usepackage{cite}
\usepackage{epsfig,graphics,graphicx,latexsym}
\usepackage{amsmath,array,amscd,amsthm,latexsym,amssymb,mathrsfs,syntonly,eucal}
\usepackage{cases, longtable}
\usepackage[usenames]{color}
\usepackage{url}

\newtheorem{thm}{Theorem}
\newtheorem{proposition}{Proposition}

\newtheorem{remark}[thm]{Remark}
\newtheorem{definition}[thm]{Definition}

\begin{document}

\title{Interference Two-Way Relay Channel with Three End-nodes}

\author{
\IEEEauthorblockN{Erhan Y\i{}lmaz and Raymond Knopp and David Gesbert}
\IEEEauthorblockA{EURECOM, Sophia-Antipolis France\\
\{yilmaz, knopp, gesbert\}@eurecom.fr }}

\maketitle

\begin{abstract}
In this paper, we study a communication system consisting of three end-nodes, e.g. a single transceiver base station (BS), one transmitting and one receiving user equipments (UEs), and a common two-way relay node (RN) wherein the full-duplex BS transmits to the receiving UE in downlink direction and receives from the transmitting UE in uplink direction with the help of the intermediate full-duplex RN. We call this system model as interference two-way relay channel (ITWRC) with three end-nodes. Information theoretic bounds corresponding this system model are derived and analyzed so as to better understand the potentials of exploiting RN in future communication systems. Specifically, achievable rate regions corresponding to decode-and-forward (DF) relaying with and without rate splitting, and partial-DF and compress-and-forward (pDF+CF) relaying strategies are derived.
\end{abstract}

\section{Introduction \label{sec:intro}}

In this paper, we study the system model illustrated in Fig.~\ref{fig:TwoWayRelayInCellularSystems} wherein a common full-duplex base station (BS) communicates simultaneously with two distinct user equipments (UEs) in uplink (UL) and downlink (DL) directions with the help of an intermediate full-duplex relay node (RN). Due to its tangential relation to the interference relay channel (IRC) \cite{maric/debora_isit08} and the two-way relay channel (TWRC) \cite{rankov_asilomar05, knopp_izs06}, we call this channel model as interference two-way relay channel (ITWRC) with three end-nodes. However, this channel model has several distinct characteristics which differentiate it from the IRC and the TWRC.

In conventional TWRC, where two UEs exchange information via an intermediate RN \cite{rankov_asilomar05, knopp_izs06, wu/chou_ciss05, popovski_icc06, kim/mitran_conf07, katti_sigcomm07, gunduz_allerton08}, the communicating UEs first send their messages to the RN which then processes the received signals according to a given relaying strategy and broadcasts to the mobiles. Two-way relaying provides interference-free reception at each mobile by canceling the self-interference before decoding the unknown message. Though there are some similarities, our system model differs from the conventional TWRCs in two main points. First, in our model instead of having two nodes communicating with each other via a RN, two UEs communicate with a BS, one in UL direction and the other in DL direction, with the aid of an intermediate RN. Secondly, if the UEs are close to each other then the receiving UE will see the \emph{interference} from the other UE both from the direct path and through the RN path. Hence, the interference needs to be treated in a subtle way in this setup.

In IRC, a common RN helps simultaneously two (or more) source-destination pairs where each source creates interference to non-intended destinations \cite{maric/debora_isit08, sahin/simeone_isit09}. Depending on channel conditions, the RN can be exploited for both cooperative signal and interference forwarding purposes to help destinations. As opposed to the IRC, in our system model there are three end-nodes where one of the nodes (i.e. the BS) is both transmitter and receiver. Hence, this transceiver node might exploit its transmit signal as side information while processing its received signal coming from the RN.

In this paper, we study the discrete memoryless ITWRC with three end-nodes which might represent a cellular network where UL and DL communications are multiplexed by using a full-duplex RN. For this system model, we explore two relaying strategies which exploit characteristics of TWRC and interference channels (ICs) in order to attain better achievable rates. In particular, a cut-set outer bound and two achievable rate regions corresponding to decode-and-forward (DF) relaying with and without rate splitting at the transmitting UE, and partial DF and compress-and-forward (pDF+CF) relaying strategies are derived.

\begin{figure}[t]
\centering
\includegraphics[width=3.3in, height=2.2in, viewport = 20 50 700 530, clip]{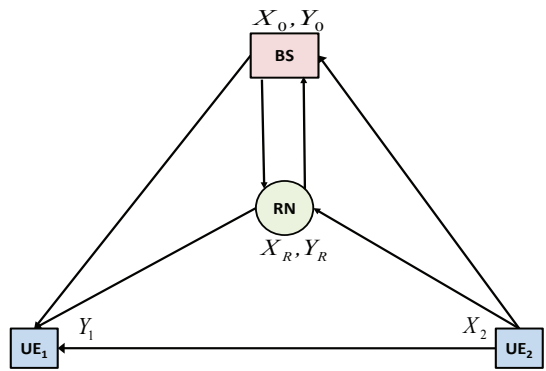}
\caption{A general Interference Two-Way Relay Channel (ITWRC) with three end-nodes in cellular set-up. \label{fig:TwoWayRelayInCellularSystems}}
\end{figure}


\section{The channel model \label{sec:DMCmodel}}

Consider the ITWRC scenario illustrated in Fig.~\ref{fig:TwoWayRelayInCellularSystems}. Here, a common full-duplex BS communicates simultaneously with two distinct UEs in UL and DL directions with the help of an intermediate RN. In the following, we first give the discrete memoryless (d.m.) ITWRC model shown in Fig.~\ref{fig:dmITWRC}.

\begin{definition}
The d.m. ITWRC with three end-nodes is defined by $\big\{(\mathcal{X}_0, \mathcal{X}_2, \mathcal{X}_R)$, $p(y_0, y_1, y_R | x_0, x_2, x_R)$, $(\mathcal{Y}_0, \mathcal{Y}_1, \mathcal{Y}_R) \big\}$, where $\mathcal{X}_0$, $\mathcal{X}_2$ and $\mathcal{X}_R$ are the input alphabets, $\mathcal{Y}_0$, $\mathcal{Y}_1$ and $\mathcal{Y}_R$ are the output alphabets and $p(.|.)$ is the channel transition probability matrix. All the alphabets are finite. The time invariant and memoryless channel is represented by
\begin{align*}
\Pr(\textbf{y}_0, \textbf{y}_1, \textbf{y}_R | \textbf{x}_0, \textbf{x}_2, \textbf{x}_R) = \prod_{i=1}^{n} p(y_{0i}, y_{1i}, y_{Ri} | x_{0i}, x_{2i}, x_{Ri})
\end{align*}
where $x_{0i}, x_{2i}, x_{Ri}, y_{0i}, y_{1i}, y_{Ri}$ are the inputs and outputs of the channel at time $i$, respectively.
\end{definition}

At the beginning of each block of $n$ channel uses, the message sources $\textrm{BS}$ and $\textrm{UE}_2$ produce random integers $W_0$ and $W_2$ from the sets $\mathcal{W}_0 \in \{1, 2, \ldots, 2^{n R_0}\}$ and $\mathcal{W}_2 \in \{1, 2, \ldots, 2^{n R_2}\}$, respectively. The message pair $(W_0, W_2)$ is drawn according to a uniform distribution over $\mathcal{W}_0 \times \mathcal{W}_2$ and occurs with probability $1 / 2^{n (R_0 + R_2)}$.

We assume restricted encoders at the $\textrm{BS}$ and $\textrm{UE}_2$, e.g., the encoder outputs do not depend on feedback signals. Hence, we define an $((2^{n R_0}, 2^{n R_2}), n)$ code for the ITWRC as follows:
\begin{itemize}
	\item An index $W_m$, for each transmitter terminal, selected uniformly from the message set $\mathcal{W}_m$ and the corresponding codeword $X^{n}_{m} (W_m) \in \mathcal{X}_m ,\; m \in \{0, 2\}$
	\item Two source encoding functions that map each message $W_m \in \mathcal{W}_m$ into a codeword $\mathcal{X}_m^n (W_m)$
	\begin{align*}
	  f_{\textrm{BS}}   &: \mathcal{W}_0 \rightarrow \mathcal{X}_0^n, \\
	  f_{\textrm{UE}_2} &: \mathcal{W}_2 \rightarrow \mathcal{X}_2^n,
	\end{align*}	
	\item A set of relay encoding (causal) functions $\{f_{R, i}\}_{i=1}^{n}$ such that
	\begin{align*}
	  x_{R,i} &= f_{R, i}(Y_{R,1}, Y_{R,2}, \ldots, Y_{R,i-1}), \quad 1 \leq i \leq n,
	\end{align*}	
	\item and two decoding functions at the $\textrm{BS}$ and $\textrm{UE}_1$
	\begin{align*}
	  g_{\textrm{BS}}   &: \mathcal{Y}_{0}^n \rightarrow \mathcal{W}_{2} \\
	  g_{\textrm{UE}_1} &: \mathcal{Y}_{1}^n \rightarrow \mathcal{W}_{0}.
	\end{align*}	
\end{itemize}

We also define the average probability of error as
\begin{align*}
P_{e}^{(n)} &= \Pr\left[ \{W_0 \neq g_{\textrm{UE}_1}(\mathcal{Y}_{1}^n)\}  \cup \{W_2 \neq g_{\textrm{BS}}(\mathcal{Y}_{0}^n)\} \right].
\end{align*}	

A rate pair $(R_0, R_2)$ is said to be achievable for the ITWRC if there exists a sequence of $(2^{n R_0}, 2^{n R_2}, n)$ codes with the average probability of error $P_{e}^{(n)} \rightarrow 0$ as $n \rightarrow \infty$. The capacity region is the closure of the set of all achievable rate pairs $(R_0, R_2)$.

\section{Outer Bound \label{sec:OuterBound}}

Considering Fig.~\ref{fig:TwoWayRelayInCellularSystems} where all the nodes can see each other, using cut-set theorem \cite{book:Cover} for the DL communication from the BS to the $\textrm{UE}_1$ and for the UL communication from the $\textrm{UE}_2$ to the BS we have the following outer region:

\begin{subequations}  \label{eq:FDUB}
{\footnotesize
\begin{align}
R_{DL} &= \min \left\{\mathrm{I}\left(X_{0};Y_R, Y_1|X_{R}, X_{2}\right), \mathrm{I}\left(X_0, X_R; Y_1 | X_2 \right)\right\}  \label{eq:FDUB_Rdl} \\
R_{UL} &= \min \left\{\mathrm{I}\left(X_{2};Y_R, Y_0|X_{R}, X_{0}\right), \mathrm{I}\left(X_2, X_R; Y_0 | X_0 \right)\right\} \label{eq:FDUB_Rul}
\end{align}}
\end{subequations}
for a joint distribution that factors as
\begin{align} \label{eq:ITWRC:ULDL:UB:pdf}
p(x_0, x_2, x_R) p(y_0, y_1, y_R |x_0, x_2, x_R)
\end{align}	
where $X_0, X_R, X_2$ are the transmit signals at the BS, RN and $\textrm{UE}_2$, and $Y_0, Y_R, Y_1$ are the received signals at the BS, RN and $\textrm{UE}_1$, respectively.

\begin{remark}
Note that if we assume $\textrm{UE}_1$ and $\textrm{UE}_2$ are merged to be a single UE, e.g. $(X_1, Y_1) = (X_2, Y_2)$ (which is possible if we assume an infinite capacity link between the two UEs), and no direct link between $\textrm{UE}_2$ and the BS, e.g. $X_2 \leftrightarrow (X_R, Y_R) \leftrightarrow Y_0$ and $X_0 \leftrightarrow (X_R, Y_R) \leftrightarrow Y_2$ form a Markov chain. This particular setup is the TWRC with no direct link between the transmitters, the BS and UE, for which the outer bound given in \eqref{eq:FDUB} matches to the one derived in \cite{gunduz_allerton08}.
\end{remark}

\section{Decode-and-Forward (DF) relaying \label{sec:DF}}

\subsubsection{DF Relaying without Rate-Splitting \label{subsubsec:FD_DFNoRateSplit}}

In the following, we give an achievable rate region corresponding to a DF relaying scheme where no \emph{rate-splitting} employed at the $\textrm{UE}_2$ (note that in any case, for the channel model in consideration, rate-splitting is not needed at the $\textrm{BS}$), and $\textrm{UE}_1$ tries to decode all transmitted messages, first proposed in \cite{maric/debora_asilomar08} for the IRC.

\begin{proposition} \label{pro:FDDF:NoRateSplit}
For the system model defined above, any rate pair $(R_0, R_2)$ that satisfies
{\footnotesize
\begin{subequations} \label{eq:FDDF:NoRateSplit}
\begin{align}
&R_0 \leq \min \{ \mathrm{I}\left( X_0, X_R ; Y_1 | U_2, X_2 \right), \mathrm{I}\left( X_0 ; Y_R | X_2, U_0, U_2 \right)\}  \label{eq:FDDF:NoRateSplit_1} \\
&R_2 \leq \min \{ \mathrm{I}\left( X_2, X_R ; Y_0 | U_0, X_0 \right), \mathrm{I}\left( X_2 ; Y_R | X_0, U_0, U_2 \right)\}  \label{eq:FDDF:NoRateSplit_2} \\
&R_0 + R_2 \leq \min \{ \mathrm{I}\left( X_0, X_2, X_R ; Y_1 \right), \mathrm{I}\left( X_0, X_2 ; Y_R | U_0, U_2 \right)\}  \label{eq:FDDF:NoRateSplit_3}
\end{align}
\end{subequations}}
for a joint distribution that factors as
\begin{align*}
p(u_0, x_0) p(u_2, x_2) f(x_R | u_0, u_2) p(y_{0}, y_{1}, y_{R} | x_{0}, x_{2}, x_{R}),
\end{align*}
where $f(.)$ is a deterministic function, is achievable by using DF relaying without rate-splitting at the $\textrm{UE}_2$.
\end{proposition}
\begin{proof}
The proof follows from \cite[Theorem 1]{maric/debora_asilomar08} with proper definition of random variables.
\end{proof}

\begin{remark}
The second terms of bounds \eqref{eq:FDDF:NoRateSplit_1}~-~\eqref{eq:FDDF:NoRateSplit_3} are required in order to provide reliable decoding at the relay. Since the RN decodes both indexes, possible error events at the RN are the same as in the multi access channel (MAC) \cite{book:Cover}. The first terms of bounds \eqref{eq:FDDF:NoRateSplit_1}~-~\eqref{eq:FDDF:NoRateSplit_3} are due to decoding constraints at the $\textrm{BS}$ and $\textrm{UE}_1$. Note that, for this particular scheme, in the encoding of the $\textrm{UE}_2$'s messages rate-splitting is not exploited. $\textrm{UE}_1$ jointly decodes messages $(W_0, W_2)$ as in the MAC. However, compared to the regular MAC rate constraints, the error in decoding the unwanted message (sent by $\textrm{UE}_2$) at $\textrm{UE}_1$ is ignored, as in \cite{maric/debora_asilomar08}, and therefore there is one less rate constraint when compared to the MAC rate bounds. Note also that the channel seen by $\textrm{BS}$ is equivalent to the regular relay channel model \cite{cover_relay_jnl79}, and since it has its own message, it can cancel its own interference from the signal forwarded by the RN.
\end{remark}

\subsubsection{DF Relaying with Rate-Splitting \label{subsubsec:FD_DFwithRateSplit}}

Rate splitting is known to be the best achievable scheme \cite{han/kobayashi_jnl81} for the ICs. In our system model, since the $\textrm{UE}_2$ interferes with $\textrm{UE}_1$ we also exploit \emph{rate splitting} at the $\textrm{UE}_2$ (note that rate-splitting is not needed at the $\textrm{BS}$ due to having side information) for DF relaying. For the ITWRC with three end-nodes, we have the following theorem on achievable rates corresponding to the DF with rate-splitting.

\begin{thm} \label{theorem:RateSplitting}
For the d.m. ITWRC, the following rate region, $\mathcal{R}_{DF+RS}$,
{\footnotesize
\begin{subequations} \label{eq:ITWRC:ULDL:DF:RN}
\begin{align}
&\mathcal{R}_{DF+RS} \stackrel{\Delta}{=} \Bigg\{(R_0, R_2): R_0\geq 0, R_{2} \geq 0,   \notag \\
&R_{0}        \leq \min \Big\{ \mathrm{I}\left( X_0 ; Y_R | U_0, U_{2c}, U_{2p}, X_{2c}, X_{2p} \right) , \notag \\
&\hspace{1.6in}  \mathrm{I}\left(X_{0}, X_R ; Y_1 | U_{2c}, X_{2c} \right) \Big\}    \label{eq:ITWRC:ULDL:DFRN:2a} \\
&R_{2}        \leq \min \Big\{ \mathrm{I}\left( X_{2c}, X_{2p}; Y_R | U_0, U_{2c}, U_{2p}, X_0  \right),  \notag \\
&\hspace{0.7in}  \mathrm{I}\left(X_{2c}, X_{2p}, X_R ; Y_0 | U_0, X_0 \right),  \notag \\
&\hspace{1in}  \mathrm{I}\left(X_{2p} ; Y_R | U_0, U_{2c}, U_{2p}, X_0, X_{2c} \right) \notag \\
&\hspace{1in}  + \min \big\{ \mathrm{I}\left(X_{2c}; Y_R | U_0, U_{2c}, U_{2p}, X_0, X_{2p} \right), \notag \\
&\hspace{1.6in} \mathrm{I}\left(X_{2c}, X_R ; Y_1| U_0, X_0 \right)\big\} \Big\}    \label{eq:ITWRC:ULDL:DFRN:2b} \\
&R_0 + R_{2}   \leq \displaystyle \min \Big\{ \mathrm{I}\left( X_0, X_{2c}, X_{2p} ; Y_R | U_0, U_{2c}, U_{2p} \right),  \notag \\
&\hspace{0.1in} \displaystyle \mathrm{I}\left(X_{2p} ; Y_R | U_0, U_{2c}, U_{2p}, X_0, X_{2c} \right) +  \notag \\
&\hspace{0.1in} \min \big\{ \mathrm{I}\left(X_0, X_{2c}; Y_R | U_0, U_{2c}, U_{2p}, X_{2p} \right), \mathrm{I}\left(X_{0}, X_{2c}, X_R ; Y_1 \right)\big\}, \notag \\
&\hspace{0.1in} \displaystyle \mathrm{I}\left( X_0, X_{2p}; Y_R | U_0, U_{2c}, U_{2p}, X_{2c} \right) +  \notag \\
&\hspace{0.1in} \min \big\{\mathrm{I}\left( X_{2c} ; Y_R | U_0, U_{2c}, U_{2p}, X_0, X_{2p} \right), \mathrm{I}\left(X_{2c}, X_R ; Y_1 | U_0, X_0 \right) \big\}  \Big\}   \label{eq:ITWRC:ULDL:DFRN:2c}\\
&2 R_0 + R_{2} \leq \displaystyle \mathrm{I}\left( X_0, X_{2p}; Y_R | U_0, U_{2c}, U_{2p}, X_{2c} \right) + \notag \\
&\hspace{0.1in} \min\big\{\mathrm{I}\left( X_0, X_{2c}; Y_R | U_0, U_{2c}, U_{2p}, X_{2p} \right), \mathrm{I}\left(X_{0}, X_{2c}, X_R ; Y_1 \right) \big\} \Bigg\} \label{eq:ITWRC:ULDL:DFRN:2d}
\end{align}	
\end{subequations}}
is achievable with the DF relaying and rate splitting at the $\textrm{UE}_2$ for a joint distribution that factors as
{\footnotesize
\begin{align} \label{eq:ITWRC:ULDL:DF:pdf}
p(x_0) p(u_2) p(x_2 | u_2) p(x_R | x_0, u_2) p(y_0, y_1, y_R |x_0, x_2, x_R).
\end{align}}	
\end{thm}

\begin{proof}
See Appendix~\ref{app:DFproof} for the proof.
\end{proof}

\begin{figure}[t]
\centering
\includegraphics[width=3.3in, height=2.2in]{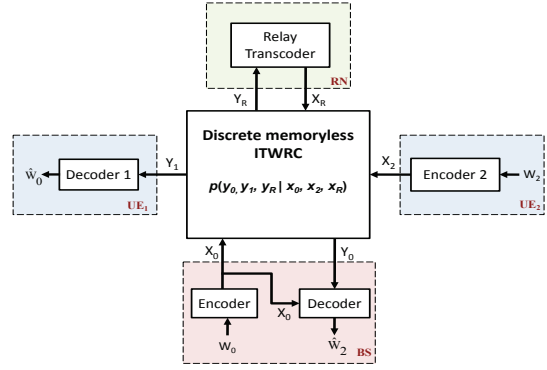}
\caption{The discrete memoryless ITWRC with three end-nodes. \label{fig:dmITWRC}}
\end{figure}

\section{Mixed Partial DF and CF (pDF+CF) relaying \label{sec:pDF+CF}}

In this scheme, we give a coding scheme where the RN partially decodes a part of messages transmitted by each transmitter; and compresses the remaining part using Wyner-Ziv compression \cite{wyner_jnl76}. To obtain an achievable rate region we use backward decoding followed by sliding-window decoding technique at both destinations (BS and $\mathrm{UE}_1$). Here, the $\mathrm{UE}_1$ also decodes some part of the message transmitted by the interfering user $\mathrm{UE}_2$, which is also decoded at the RN, so as to mitigate interference seen at the $\mathrm{UE}_1$. For the ITWRC with three end-nodes, assuming partial DF and CF (pDF+CF) relaying, we have the following corresponding theorem.

\begin{thm} \label{theorem:pDF+CF}
For the d.m. ITWRC with three end-nodes, using partial DF and CF relaying the rate region, $\mathcal{R}_{pDF+CF}$,
{\footnotesize
\begin{subequations} \label{eq:ITWRC:ULDL:pDF+CF}
\begin{align}
&\mathcal{R}_{pDF+CF} \stackrel{\Delta}{=} \Bigg\{(R_0, R_2) : R_0 \geq 0, R_{2} \geq 0,   \label{eq:ITWRC:ULDL:cor:pDF+CF:a} \\
&R_{0} \leq I(X_0; Y_1, \hat{Y}_R \;|\; U_0, U_2, V_0, V_2, X_R) \notag \\
&\hspace{0.3in} + \min \big\{I(V_0 ; Y_R | X_R, U_0, U_2, V_2), I(U_0, V_0 ; Y_1 | U_2, V_2) \big\}  \label{eq:ITWRC:ULDL:cor:pDF+CF:b} \\
&R_{2} \leq  I(X_2; Y_0, \hat{Y}_R \;|\; U_0, U_2, V_0, V_2, X_0, X_R) \notag \\
&\hspace{0.3in} + \min \big\{I(V_2 ; Y_R | X_R, U_0, U_2, V_0), I(U_2, V_2 ; Y_0 | U_0, V_0, X_0), \notag \\
&\hspace{1in} I(U_2, V_2 ; Y_1 | U_0, V_0) \big\} \label{eq:ITWRC:ULDL:cor:pDF+CF:c} \\
&R_0 + R_{2} \leq I(X_0; Y_1, \hat{Y}_R \;|\; U_0, U_2, V_0, V_2, X_R) \notag \\
&\hspace{0.6in} + I(X_2; Y_0, \hat{Y}_R \;|\; U_0, U_2, V_0, V_2, X_0, X_R) \notag \\
&\hspace{0.6in} + \min \big\{I(V_0, V_2 ; Y_R \;|\; X_R, U_0, U_2), \notag \\
&\hspace{1.2in} I(U_0, U_2, V_0, V_2 ; Y_1) \big\} \Bigg\} \label{eq:ITWRC:ULDL:cor:pDF+CF:d}
\end{align}	
\end{subequations}}
subject to constraint
{\footnotesize
\begin{align} \label{eq:rateConstraint:cor:pDF+CF}
&\max \big\{\mathrm{I}\left(\hat{Y}_R ; Y_R | Y_0, X_0, X_R, V_0, V_2 \right), \mathrm{I}\left(\hat{Y}_R ; Y_R | Y_1, X_R, V_0, V_2\right) \big\}  \notag \\
&\;\; \leq \min \big\{\mathrm{I}\left(X_R ; Y_0 | U_0, U_2, V_0, V_2, X_0 \right), \mathrm{I}\left(X_R ; Y_1 | U_0, U_2, V_0, V_2 \right) \big\}
\end{align}}
is achievable by using mixed pDF+CF relaying for a joint distribution that factors as
\begin{align} \label{eq:ITWRC:ULDL:pDF+CF:pdf}
&p(u_0) p(v_0 | u_0) p(x_0 | v_0) p(u_2) p(v_2 | u_2) p(x_2 | v_2) p(x_R | u_0, u_2) \notag \\
& \hspace{0.1in} \cdot p(\hat{y}_R | y_R, x_R, u_0, u_2, v_0, v_2) p(y_0, y_1, y_R |x_0, x_2, x_R).
\end{align}
\end{thm}

\begin{proof}
See Appendix~\ref{app:pDF+CFproof} for the proof.
\end{proof}

\begin{remark}
In Theorem~\ref{theorem:pDF+CF} if we set $U_0 = U_2 = V_0 = V_2 \equiv \emptyset$ then the obtained achievable rate region corresponds to pure compress-and-forward (CF)
relaying strategy.
\end{remark}

\section{Discussions \label{sec:discussions}}

Although we are not able to share numerical results for Gaussian channel case, due to space limitation, in the following we want to add some comments based on our observations.

In the current cellular systems, UL and DL communications are orthogonal either in time (time-division duplex) or in frequency (frequency-division duplex). One of the major problems in these conventional duplexing schemes appears in the UL communications from UE to BS since limited power resources at UEs. To tackle this problem, i.e. to extend cell coverage, recently RN deployment in 4G cellular systems (such as LTE-A and 802.16j) has been proposed.

Regarding the framing structures of 4G cellular systems, see \cite{lte_211} for LTE, it is viable to have better achievable rates than the conventional schemes by multiplexing UL and DL communication via a two-way RN, as in our ITWRC model.
However, due to asymmetric channel conditions in UL and DL directions, selection of the forwarding strategy plays a cardinal role in better harnessing the potentials of two-way relaying. Regarding transmit power and channel asymmetries, the DF relaying might provide better achievable rates since the multiplexing loss in decoding both BS and UE signals will vanish with the difference in received signal powers.

\section{Conclusions \label{sec:conclusion}}

In this paper, we studied a three-node interference two-way relay channel consisting of one BS, two UEs  and a single RN. Communicating two information flows, one in uplink and the other in downlink direction, between the BS and the UEs are facilitated by the intermediate two-way RN. Specifically, we derived a cut-set outer bound and two achievable rate regions corresponding to the DF relaying with and without rate splitting, and the partial DF and CF relaying strategies. We also pointed out that a two-way relay node, which enables concurrent uplink and downlink communications, might improve spectral efficiency in the next generation cellular systems.

\begin{appendices}

\section{Proof of Theorem \ref{theorem:RateSplitting}} \label{app:DFproof}

Fix the input distribution $p(x_0) p(u_2) p(x_2 | u_2) p(x_R | x_0, u_2)$. The messages $W_{0,b} \in [1, 2^{n R_{0}}]$ and $W_{2,b} = (W_{2c,b}, W_{2p,b}) \in [1, 2^{n R_{2}}]$ where $W_{2c,b} \in [1, 2^{n R_{2c}}]$ and $W_{2p,b} \in [1, 2^{n R_{2p}}]$ with $R_2 = R_{2c} + R_{2p}$, for $b = 1, 2, \ldots, B$, will be sent over the ITWRC in $B+1$ blocks each of $n$ transmissions. The messages are uniformly distributed and independent of each other. A random coding argument is used to show the achievability of $\mathcal{R}_{DF+RS}$. In each of the $B$ blocks the same codebook is used, i.e. regular encoding. Note that as $B \rightarrow \infty$, $R_t B /(B+1) \rightarrow R_t, \; t \in \{0, 2c, 2p\}$. 
In the following all the codewords are assumed to be vector of length $n$.

\textbf{\textit{Codebook construction}:}
\begin{itemize}
	\item Generate $2^{n R_0}$ i.i.d. sequences $\textbf{u}_0$ each with probability $p(\textbf{u}_0) = \prod_{i=1}^{n} p(u_{0,i})$. Label them $\textbf{u}_0 (w'_0)$, where $w'_{0} \in [1, 2^{n R_0}]$.
	\item For every $\textbf{u}_0(w'_0)$ generate $2^{n R_0}$ i.i.d. sequences $\textbf{x}_0$ each with probability $p(\textbf{x}_0 | \textbf{u}_1 (w'_0)) = \prod_{i=1}^{n} p(x_{0,i} | u_{0,i}(w'_0))$. Label them $\textbf{x}_0(w'_0, w_0)$, where $w_{0} \in [1, 2^{n R_0}]$.
	\item Generate $2^{n R_{2c}}$ i.i.d. sequences $\textbf{u}_{2c}$ and $2^{n R_{2p}}$ sequences $\textbf{u}_{2p}$ each with probability $\prod_{i=1}^{n} p(u_{2c, i})$ and $\prod_{i=1}^{n} p(u_{2p, i})$, respectively. Label them $\textbf{u}_{2c}(w'_{2c})$ and $\textbf{u}_{2p}(w'_{2p})$ where $w'_{2c} \in [1, 2^{n R_{2c}}]$ and $w'_{2p} \in [1, 2^{n R_{2p}}]$.
	\item For every $\textbf{u}_{2c}(w'_{2c})$ generate $2^{n R_{2c}}$ i.i.d. sequences $\textbf{x}_{2c}$ each with probability $\prod_{i=1}^{n} p(x_{2c,i} | u_{2c,i}(w'_{2c}))$. Label them $\textbf{x}_{2c}(w'_{2c}, w_{2c})$, where $w_{2c} \in [1, 2^{n R_{2c}}]$.
	\item For every $\textbf{u}_{2p}(w'_{2p})$ generate $2^{n R_{2p}}$ i.i.d. sequences $\textbf{x}_{2p}$ each with probability $\prod_{i=1}^{n} p(x_{2p,i} | u_{2p,i}(w'_{2p}))$. Label them $\textbf{x}_{2p}(w'_{2p}, w_{2p})$, where $w_{2p} \in [1, 2^{n R_{2p}}]$.
	\item For each $(w'_{2c}, w_{2c}, w'_{2p}, w_{2p})$ generate $\textbf{x}_{2}(w'_{2c}, w_{2c}, w'_{2p}, w_{2p})$ which is a deterministic function of $(\textbf{x}_{2c}, \textbf{x}_{2p})$.
	\item For each $(w'_{0}, w'_{2c}, w'_{2p})$ generate $\textbf{x}_{R}(w'_{0}, w'_{2c}, w'_{2p})$ which is a deterministic function of $(\textbf{u}_0, \textbf{u}_{2c}, \textbf{u}_{2p})$.
\end{itemize}

\textbf{\textit{Encoding}:} The messages $w_{0,b}$, $w_{2c,b}$, $w_{2p,b}$, $b = 1, \ldots, B$, are encoded in superposition block Markov fashion where in the first block, $b=1$, the $\textrm{BS}$ sends $\textbf{x}_0(1, w_{0,1})$ and the $\textrm{UE}_2$ sends $\textbf{x}_{2}(1, w_{2c, 1}, 1, w_{2p,1})$ which is a deterministic function of both $\textbf{x}_{2c}(1, w_{2c,1})$ and $\textbf{x}_{2p}(1, w_{2p,1})$; and the $\textrm{RN}$ transmits $\textbf{x}_{R}(1, 1, 1)$. Then, in the block $b$, $b = 2, 3, \ldots, B$
\begin{align*}
\textbf{x}_{0,b} &= \textbf{x}_1(w'_{0, b}, w_{0, b}), \\
\textbf{x}_{2,b} &= \textbf{x}_2(w'_{2c,b}, w_{2c,b}, w'_{2p,b}, w_{2p,b})
\end{align*}
where $(w'_{0, b}, w'_{2c, b},w'_{2p, b}) = (w_{0,b-1}, w_{2c,b-1}, w_{2p,b-1})$. Assume that before block $b$, $b = 2, 3, \ldots, B+1$, the $\textrm{RN}$ has the estimates $(\hat{w}_{0,b-1}, \hat{w}_{2c, b-1}, \hat{w}_{2p, b-1})$ for the message triple $(w_{0, b-1}, w_{2c, b-1}, w_{2p, b-1})$; and hence it transmits $\textbf{x}_{r}(w'_{0, b}, w'_{2c, b}, w'_{2p, b})$ where $(w'_{0, b}, w'_{2c, b}, w'_{2p, b}) = (\hat{w}_{0, b-1}, \hat{w}_{2c, b-1}, \hat{w}_{2p, b-1})$. And in block $B+1$ the $\textrm{BS}$ and $\textrm{UE}_2$ transmit, respectively,
\begin{align*}
\textbf{x}_{0,B+1} &= \textbf{x}_1(w_{0,  B}, 1), \\
\textbf{x}_{2,B+1} &= \textbf{x}_2(w_{2c, B}, 1, w_{2p, B}, 1).
\end{align*}

\textbf{\textit{Decoding at the $\textrm{RN}$}:} For the decoding at the $\textrm{RN}$ in order to obtain \emph{cooperation} after each block $b$, $b = 1, 2, \ldots, B$, the $\textrm{RN}$ chooses $(\hat{w}_{0, b}, \hat{w}_{2c, b}, \hat{w}_{2p, b})$, assuming it has already decoded the previous message triplet $(w_{0, b-1}, w_{2c, b-1}, w_{2p, b-1})$ correctly, such that
\begin{align} \label{eq:TWRC:ULDL:DFRN:TypicalDecoder:Relay}
&\displaystyle \{\textbf{u}_{0}(w'_{0,b}), \textbf{x}_0 (w'_{0,b}, \hat{w}_{0, b}), \textbf{u}_{2c}(w'_{2c, b}), \textbf{x}_{2c} (w'_{2c, b}, \hat{w}_{2c, b}), \notag \\
&\displaystyle \hspace{0.1in} \textbf{u}_{2p}(w'_{2p, b}), \textbf{x}_{2p} (w'_{2p, b}, \hat{w}_{2p, b}), \textbf{x}_{R} (w'_{0, b}, w'_{2c, b}, w'_{2p, b}), \textbf{y}_{R, b}\} \notag \\
&\displaystyle \hspace{0.6in} \in \mathcal{A}_{\epsilon}(U_0, X_0, U_{2c}, X_{2c}, U_{2p}, X_{2p}, X_R, Y_R)
\end{align}
where $ \mathcal{A}_{\epsilon}(\cdot)$ represents $\epsilon$-strongly typical sets \cite{book:Cover}.

The error analysis at the $\textrm{RN}$ corresponds to MAC with three users \cite{book:Cover}. From \eqref{eq:TWRC:ULDL:DFRN:TypicalDecoder:Relay} if the following rate constraints
\begin{subequations} \label{eq:DFRN:ratesAtRN}
\begin{align}
\sum_{i \in \mathcal{S}} R_{i} &\leq \mathrm{I}\left( X(\mathcal{S}); Y_R | U_0, U_{2c}, U_{2p}, X(\mathcal{S}^c) \right)
\end{align}	
\end{subequations}
for all $\mathcal{S} \subseteq \{0, 2c, 2p\}$ where $X(\mathcal{S}) = \{X_i : i \in \mathcal{S}\}$ and $n \rightarrow \infty $ are satisfied then the decoding error probability can be made small.

\textbf{\textit{Decoding at the $\textrm{BS}$ and $\textrm{UE}_1$}:} For the decoding process a \emph{backward decoding} scheme is used at both the $\textrm{BS}$ and $\textrm{UE}_1$. We start with the decoding process at the $\textrm{BS}$. First we note that the $\textrm{BS}$ does not suffer from interference, since it can precancel its own transmitted signal before starting the decoding, i.e., it has side information of its own transmitted signal.

Then in block $b$ the $\textrm{BS}$ looks for the pair $(\hat{w}'_{2c, b}, \hat{w}'_{2p, b}) = (\hat{w}_{2c, b-1}, \hat{w}_{2p, b-1})$, assuming it has already decoded the \emph{future} message pair $(w_{2c, b}, w_{2p, b})$ correctly, such that
\begin{align} \label{eq:TWRC:ULDL:DFRN:TypicalDecoder:BS}
&\displaystyle \{\textbf{u}_{2c}(\hat{w}'_{2c, b}), \textbf{x}_{2c} (\hat{w}'_{2c, b}, w_{2c, b}), \textbf{u}_{2p}(\hat{w}'_{2p, b}), \textbf{x}_{2p} (\hat{w}'_{2p, b}, w_{2p, b}), \notag \\
&\displaystyle \hspace{0.1in} \textbf{x}_{R} (w'_{0, b}, \hat{w}'_{2c, b}, \hat{w}'_{2p, b}), \textbf{u}_{0}(w'_{0, b}), \textbf{x}_0 (w'_{0, b}, w_{0, b}),  \textbf{y}_{0, b}\} \notag \\
&\hspace{0.6in}\displaystyle \in \mathcal{A}_{\epsilon}(U_0, X_0, U_{2c}, X_{2c}, U_{2p}, X_{2p}, X_R, Y_0).
\end{align}

The error analysis at the $\textrm{BS}$ corresponds to single-user decoding \cite{book:Cover}. From \eqref{eq:TWRC:ULDL:DFRN:TypicalDecoder:BS}, if the following rate constraints
\begin{align}
R_{2c} + R_{2p} &\leq \mathrm{I}\left(U_{2c}, U_{2p}, X_{2c}, X_{2p}, X_R ; Y_0 | U_0, X_0 \right) \label{eq:DFRN:ratesAtBS}
\end{align}	
and $n \rightarrow \infty $ are satisfied then the decoding error probability can be made small.

Similar to the decoding steps at the $\textrm{BS}$, the $\textrm{UE}_1$ uses \emph{backward} decoding. In block $b$, the $\textrm{UE}_1$ looks for the pair $(\hat{w}'_{0, b}, \hat{w}'_{2c, b}) = (\hat{w}_{0, b-1}, \hat{w}_{2c, b-1})$, assuming it has already decoded the previous message pair $(w_{0, b}, w_{2c, b})$ correctly, such that
\begin{align} \label{eq:TWRC:ULDL:DFRN:TypicalDecoder:UE1}
&\displaystyle \{\textbf{u}_{0}(\hat{w}'_{0,b}), \textbf{x}_{0}(\hat{w}'_{0, b}, w_{0, b}), \textbf{u}_{2c}(\hat{w}'_{2c, b}), \textbf{x}_{2c}(\hat{w}'_{2c, b}, w_{2c, b}), \notag \\
&\displaystyle \textbf{x}_{R}(w'_{0, b}, \hat{w}'_{2c, b}, \hat{w}'_{2p, b}), \textbf{y}_{1, b}\} \in \mathcal{A}_{\epsilon}(U_0, X_0, U_{2c}, X_{2c}, X_R, Y_1).
\end{align}

The error analysis at the $\textrm{UE}_1$ corresponds to two-user MAC decoding \cite{book:Cover} since it decodes the $\textrm{BS}$ message $w_{0} \in [1, 2^{n R_0}]$ and the common message $w_{2c} \in [1, 2^{n R_{2c}}]$ sent by $\textrm{UE}_2$ in order to alleviate interference effect. We note that $\textrm{UE}_1$ considers the codeword corresponding to message $w_{2p}$ as noise. From \eqref{eq:TWRC:ULDL:DFRN:TypicalDecoder:UE1}, if the following rate constraints
\begin{align}
\sum_{i \in \mathcal{S}} R_{i} &\leq \mathrm{I}\left( X(\mathcal{S}), X_R; Y_1 | U(\mathcal{S}^c), X(\mathcal{S}^c) \right)  \label{eq:DFRN:ratesAtUE1}
\end{align}	
for all $\mathcal{S} \subseteq \{0, 2c\}$ and $n \rightarrow \infty $ are satisfied, then the decoding error probability can be made small. The \emph{backward} decoding proceeds according to $b = B+1, B, \ldots, 2$ for both $\textrm{BS}$ and $\textrm{UE}_1$ assuming each has decoded the corresponding messages correctly in the block $b+1$.

By combining \eqref{eq:DFRN:ratesAtRN}, \eqref{eq:DFRN:ratesAtBS} and \eqref{eq:DFRN:ratesAtUE1}, and after applying the Fourier-Motzkin elimination to remove $R_{2p}$ by replacing it with $ R_{2p} = R_2 - R_{2c} \geq 0$ we get \eqref{eq:ITWRC:ULDL:DF:RN}. This concludes the proof.

\section{Proof of Theorem \ref{theorem:pDF+CF}} \label{app:pDF+CFproof}

Fix the input distribution $p(u_0, v_0, x_0, u_2, v_2, x_2, x_R)$ = $p(u_0) p(v_0 | u_0) p(x_0 | v_0)$ $p(u_2) p(v_2 | u_2) p(x_2 | v_2) p(x_R | u_0, u_2)$ and an $\epsilon > 0$. The messages $w_{0,b} = (w_{0,c,b}, w_{0,d,b}) \in [1, 2^{n R_{0}}]$ where $w_{0,c,b} \in [1, 2^{n R_{0,c}}]$ and $w_{0,d,b} \in [1, 2^{n R_{0,d}}]$ with $R_0 = R_{0,c} + R_{0,d}$; and $w_{2,b} = (w_{2c,b}, w_{2p,b}) \in [1, 2^{n R_{2}}]$ where $w_{2,c,b} \in [1, 2^{n R_{2,c}}]$ and $w_{2,d,b} \in [1, 2^{n R_{2,d}}]$ with $R_2 = R_{2,c} + R_{2,d}$, for $b = 1, 2, \ldots, B$, will be sent over the ITWRC in $B+1$ blocks each of $n$ transmissions. The messages are uniformly distributed and independent of each other. A random coding argument is used to show the achievability of $\mathcal{R}_{pDF+CF}$. In each of the $B+1$ blocks the same codebooks are used at each transmitter, i.e., regular encoding. As $B \rightarrow \infty$, $R_t B/(B+1) \rightarrow R_t, \; t \in \{\{0,c\}, \{0,d\}, \{2,c\}, \{2,d\}\}$.

\textbf{\textit{Codebook construction}:}
\begin{itemize}
	\item Generate $2^{n R_{0c}}$ i.i.d. codewords $\textbf{u}_0$ each with probability $p(\textbf{u}_0) = \prod_{i=1}^{n} p(u_{0,i})$.
	Label them $\textbf{u}_0 (w'_{0c})$, where $w'_{0c} \in [1, 2^{n R_{0c}}]$.
	\item For each $\textbf{u}_0(w'_{0c})$ generate $2^{n R_{0c}}$ i.i.d. sequences $\textbf{v}_0$ each with probability $p(\textbf{v}_0 | \textbf{u}_0 (w'_{0c})) = \prod_{i=1}^{n} p(v_{0,i} | u_{0,i}(w'_{0c}))$. Label them $\textbf{v}_0(w'_{0c}, w_{0c})$, where $w_{0c} \in [1, 2^{n R_{0c}}]$.
	\item For each $\textbf{v}_0(w'_{0c}, w_{0c})$ generate $2^{n R_{0d}}$ i.i.d. sequences $\textbf{x}_0$ each with probability $p(\textbf{x}_0 | \textbf{v}_0 (w'_{0c}, w_{0c})) = \prod_{i=1}^{n} p(x_{0,i} | v_{0,i}(w'_{0c}, w_{0c}))$. Label them $\textbf{x}_0(w'_{0c}, w_{0c}, w_{0d})$, where $w_{0d} \in [1, 2^{n R_{0d}}]$.
	\item Generate $2^{n R_{2c}}$ i.i.d. codewords $\textbf{u}_2 \in \mathbb{C}^n$ each with probability $p(\textbf{u}_2) = \prod_{i=1}^{n} p(u_{2,i})$.
	Label them $\textbf{u}_2 (w'_{2c})$, where $w'_{2c} \in [1, 2^{n R_{2c}}]$.
	\item For each $\textbf{u}_2(w'_{2c})$ generate $2^{n R_{2c}}$ i.i.d. sequences $\textbf{v}_2$ each with probability $p(\textbf{v}_2 | \textbf{u}_2 (w'_{2c})) = \prod_{i=1}^{n} p(v_{2,i} | u_{2,i}(w'_{2c}))$. Label them $\textbf{v}_2(w'_{2c}, w_{2c})$, where $w_{2c} \in [1, 2^{n R_{2c}}]$.
	\item For each $\textbf{v}_2(w'_{2c}, w_{2c})$ generate $2^{n R_{2d}}$ i.i.d. sequences $\textbf{x}_2$ each with probability $p(\textbf{x}_2 | \textbf{v}_2 (w'_{2c}, w_{2c})) = \prod_{i=1}^{n} p(x_{2,i} | v_{2,i}(w'_{2c}, w_{2c}))$. Label them $\textbf{x}_2(w'_{2c}, w_{2c}, w_{2d})$, where $w_{2d} \in [1, 2^{n R_{2d}}]$.
	\item For every $\{\textbf{u}_0(w'_{0c}), \textbf{u}_2(w'_{2c})\}$ generate $2^{n R_{3}}$ i.i.d. $\textbf{x}_R$ sequences each with probability $\prod_{i=1}^{n}$ $p(x_{R, i} | u_{0,i}(w'_{0c}), u_{2,i}(w'_{2c}))$. Label them $\textbf{x}_R(w'_{0c}, w'_{2c}, s)$, where $s \in [1, 2^{n R_{3}}]$.	
	\item For every $\{\textbf{x}_R(w'_{0c}, w'_{2c}, s), \textbf{v}_0(w'_{0c}, w_{0c}), \textbf{v}_2(w'_{2c}, w_{2c})\}$ generate $2^{n \hat{R}_{3}}$ i.i.d. $\hat{\textbf{y}}_r$ sequences each with probability $\prod_{i=1}^{n}$ $p(\hat{y}_{R, i} | x_{R,i}(w'_{0c}, w'_{2c}, s)$, $v_{0,i}(w'_{0c}, w_{0c}),$ $v_{2,i}(w'_{2c}, w_{2c}))$. Label them $\hat{\textbf{y}}_R(z, w'_{0c}, w'_{2c}, w_{0c}, w_{2c}, s)$, where $z \in [1, 2^{n \hat{R}_{3}}]$.	
\end{itemize}

\textbf{\textit{Random Partitions}:}
Randomly partition the set $\{1, \ldots, 2^{n \hat{R}_{3}}\}$ into $2^{n R_{3}}$ cells $\mathcal{S}_{3, s}$ and index them by $s$.

\textbf{\textit{Encoding}:}
We adopt a block Markov encoding scheme where in block $b, b = 1, \ldots, B$, the BS transmits the length-$n$ codeword $\textbf{x}_0(w_{0,c, b-1}, w_{0,c, b}, w_{0,d, b})$ in order to send $w_{0,c,b} \in [1, 2^{n R_{0,c}}]$ and $w_{0,d,b} \in [1, 2^{n R_{0,d}}]$; and similarly the $\mathrm{UE}_2$ transmits the length-$n$ codeword $\textbf{x}_2(w_{2,c, b-1}, w_{2,c, b}, w_{2,d, b})$ in order to send $w_{2,c,b} \in [1, 2^{n R_{2,c}}]$ and $w_{2,d,b} \in [1, 2^{n R_{2,d}}]$. The relay transmits the codeword $\textbf{x}_R(w_{0,c,b-1}, w_{2,c,b-1}, s_b)$ where $s_b \in [1, 2^{n R_{3}}]$ is the compressed message corresponding to the bin index that $z_{b-1}$ belongs to. Moreover, we assume that
\begin{align} 
\displaystyle &\{\hat{\textbf{y}}_R(z_{b}, w'_{0,c, b}, w'_{2,c, b}, w_{0,c,b}, w_{2,c,b}, s_{b}), \textbf{y}_R(b), \notag\\
& \hspace{0.1in} \displaystyle \textbf{v}_{0}(w'_{0,c, b}, w_{0,c,b}), \textbf{v}_{2}(w'_{2,c,b}, w_{2,c,b})), \textbf{x}_R(w'_{0,c,b}, w'_{2,c,b}, s_{b})  \notag
\end{align}
is $\epsilon$-typical and that $z_{b-1} \in \mathcal{S}_{3, s_{b}}$. We note that $w'_{0,c,b} = w_{0,c,b-1}$ and $w'_{2,c,b} = w_{2,c,b-1}$.

\textbf{\textit{Decoding at the $\textrm{RN}$}:}
For the decoding at the $\textrm{RN}$ in order to obtain \emph{cooperation} after each block $b$, $b = 1, 2, \ldots, B$, the $\textrm{RN}$ chooses
$(\hat{w}_{0,c,b}, \hat{w}_{2,c,b})$, assuming it has already correctly decoded the previous message pair $(w_{0,c,b-1}, w_{2,c, b-1})$, such that
{\footnotesize
\begin{align} \label{eq:TWRC:ULDL:TypicalDecoder:Relay}
&\displaystyle \{\textbf{u}_{0}(w'_{0,b}), \textbf{v}_0 (w'_{0,c,b}, \hat{w}_{0,c,b}), \textbf{u}_{2}(w'_{2,c,b}), \textbf{v}_{2} (w'_{2,c,b}, \hat{w}_{2,c,b}), \notag \\
&\hspace{0.2in} \displaystyle \textbf{x}_{R} (w'_{0,c,b}, w'_{2,c,b}, s_{b}), \textbf{y}_{R, b}\} \in \mathcal{A}_{\epsilon}(U_0, V_0, U_2, V_2, X_R, Y_R).
\end{align}}

The error analysis at the $\textrm{RN}$ corresponds to MAC with two users \cite{book:Cover}. From \eqref{eq:TWRC:ULDL:TypicalDecoder:Relay} if the following rate constraints
\begin{align}
\sum_{i \in \mathcal{S}} R_{i,c} &\leq \mathrm{I}\left( V(\mathcal{S}); Y_R | X_R, U_0, U_2, V(\mathcal{S}^c) \right)  \label{eq:ratesAtRN:pDF+CF}
\end{align}	
for all $\mathcal{S} \subseteq \{0, 2\}$ and $n \rightarrow \infty $ are satisfied then the decoding error probability can be made small. Also, the relay after receiving $y_{R, b}$ decides that $z_{b}$ is received if $\{\hat{\textbf{y}}_R(z_{b}, w'_{0,c, b}, w'_{2,c, b}, w_{0,c,b}, w_{2,c,b}, s_{b}), \textbf{y}_{R,b}, \textbf{v}_{0}(w'_{0,c, b}, w_{0,c,b})$, $\textbf{v}_{2}(w'_{2,c,b}, w_{2,c,b}), \textbf{x}_R(w'_{0,c,b}, w'_{2,c,b}, s_{b})\}$ is jointly $\epsilon$-typical. From rate-distortion theory, there will
exist such a $z_b$ with high probability if $n$ is sufficiently large and
\begin{align}   \label{eq:R3hat}
\hat{R}_{3} &\geq \mathrm{I}\left(\hat{Y}_R ; Y_R | X_R, V_0, V_2 \right).
\end{align}

\textbf{\textit{Decoding at the $\textrm{BS}$ and $\textrm{UE}_1$}:}
For the decoding process at both the $\textrm{BS}$ and $\textrm{UE}_1$ a \emph{backward decoding} scheme is used which is then followed by \emph{sliding-window} decoding scheme. With the backward decoding each destination node decodes the message part sent cooperatively by its source and the RN; and with the sliding-window decoding technique (after peeled out the decoded parts) each destination node, by using the compressed version of the relay received signal, decodes the remained part of the transmitted message.

We start with the decoding process at the $\textrm{BS}$. First, we note that the $\textrm{BS}$ does not suffer from interference since it can pre-cancel its own transmitted signal before starting the decoding, i.e., it has side information of its own transmitted signal.

Starting with backward decoding, in block $b, b=B+1, B, \ldots, 2$, the $\textrm{BS}$ looks for the massage $\hat{\hat{w}}'_{2,c,b} = \hat{\hat{w}}_{2,c,b-1}$, assuming it has already decoded the future message $w_{2,c,b}$ correctly, such that
\begin{align} \label{eq:TWRC:ULDL:TypicalDecoder:BS}
&\displaystyle \{\textbf{u}_{2}(\hat{\hat{w}}'_{2,c,b}), \textbf{v}_{2}(\hat{\hat{w}}'_{2,c,b}, w_{2,c,b}), \textbf{y}_{0, b} \;|\; \textbf{u}_{0}(w'_{0,c,b}), \textbf{v}_{0}(w'_{0,c,b}, w_{0,c,b}), \notag \\
&\hspace{0.1in} \displaystyle \textbf{x}_0 (w'_{0,c,b}, w_{0,c,b}, w_{0,d,b})\} \in \mathcal{A}_{\epsilon}(U_2, V_2, Y_0, U_0, V_0, X_0).
\end{align}

From \eqref{eq:TWRC:ULDL:TypicalDecoder:BS}, if the following rate constraints
\begin{align}   \label{eq:ratesAtBS:pDF+CF}
R_{2,c} &\leq \mathrm{I}\left(U_2, V_2 ; Y_0 | U_0, V_0, X_0 \right) = \mathrm{I}\left(U_2, V_2 ; Y_0 | X_0 \right)
\end{align}	
and $n \rightarrow \infty $ are satisfied then the decoding error probability can be made small \cite{book:Cover}.

Similar to the decoding steps at the $\textrm{BS}$, the $\textrm{UE}_1$ performs decoding in time backward direction for $b, b=B+1, B, \ldots, 2$. However, unlike the $\textrm{BS}$, the $\textrm{UE}_1$ sees interference from the $\textrm{UE}_2$ and in order to alleviate the interference effect, it decodes some part of the interfering signal. As such, the $\textrm{UE}_1$ looks for the massage pair $(\check{w}'_{0,c,b}, \check{w}'_{2,c,b})= (\check{w}_{0,c,b-1}, \check{w}_{2,c,b-1})$, assuming it has already decoded the future message pair $(w_{0,c,b}, w_{2,c,b})$ correctly, such that
\begin{align} \label{eq:TWRC:ULDL:TypicalDecoder:UE1}
&\displaystyle \{\textbf{u}_{0}(\check{w}'_{0,c,b}), \textbf{v}_{0}(\check{w}'_{0,c,b}, w_{0,c,b}), \textbf{u}_{2}(\check{w}'_{2,c,b}), \textbf{v}_{2}(\check{w}'_{2,c,b}, w_{2,c,b}), \notag \\
&\hspace{1in} \displaystyle \textbf{y}_{1, b}\} \in \mathcal{A}_{\epsilon}(U_0, X_0, U_2, X_2, Y_1).
\end{align}
The error analysis at the $\textrm{UE}_1$ corresponds to two-user MAC joint decoding \cite{book:Cover} where from \eqref{eq:TWRC:ULDL:TypicalDecoder:UE1} if the following rate constraints
\begin{align}
\sum_{i \in \mathcal{S}} R_{i,c} &\leq \mathrm{I}\left( U(\mathcal{S}), V(\mathcal{S}); Y_1 | U(\mathcal{S}^c), V(\mathcal{S}^c) \right) \label{eq:ratesAtUE1:pDF+CF}
\end{align}	
for all $\mathcal{S} \subseteq \{0, 2\}$ and $n \rightarrow \infty $ are satisfied then the decoding error probability can be made small \cite{book:Cover}. The \emph{backward} decoding proceeds according to $b = B+1, B, \ldots, 2$ for both $\textrm{BS}$ and $\textrm{UE}_1$ assuming each has decoded the corresponding messages correctly in the block $b+1$.

After decoding all the respective messages $w_{0,c,b}, w_{2,c,b}$, for $b=B+1, B, \ldots, 2$, at the destination nodes ($\textrm{BS}$ and $\textrm{UE}_1$), they proceed with sliding-window decoding technique to decode the relay partition bin indexes $z_{b}$ and the respective massages $w_{2,d}$ and $w_{0,d}$ by using blocks $b$ and $b+1$, for $b=1, 2, \ldots, B$. At the block $b+1$, the $\textrm{BS}$ and $\textrm{UE}_1$ first decode the message $s_{b+1}$ sent by the relay using the following respective typicality checks
{\footnotesize
\begin{align} \label{eq:TWRC:ULDL:typical:BS:index}
&\Big\{\textbf{x}_{R}(w'_{0,c,b+1}, w'_{2,c,b+1}, s_{b+1}), \textbf{y}_{0, b+1} \;|\; \textbf{u}_{2}(w'_{2,c,b+1}),  \notag \\
&\textbf{v}_{2}(w'_{2,c,b+1}, w_{2,c,b+1}), \textbf{u}_{0}(w'_{0,c,b+1}), \textbf{v}_{0}(w'_{0,c,b+1}, w_{0,c,b+1}),  \notag \\
&\textbf{x}_0 (w'_{0,c,b+1}, w_{0,c,b+1}, w_{0,d,b+1}) \Big\} \in \mathcal{A}_{\epsilon}(X_R, Y_0, U_2, V_2, U_0, V_0, X_0)
\end{align}}
and
{\footnotesize
\begin{align} \label{eq:TWRC:ULDL:typical:UE1:index}
&\Big\{\textbf{x}_{R}(w'_{0,c,b+1}, w'_{2,c,b+1}, s_{b+1}), \textbf{y}_{1, b+1} \;|\; \textbf{u}_{0}(w'_{0,c,b+1}), \notag \\
&\textbf{v}_{0}(w'_{0,c,b+1}, w_{0,c,b+1}), \textbf{u}_{2}(w'_{2,c,b+1}), \textbf{v}_{2}(w'_{2,c,b+1}, w_{2,c,b+1}) \Big\} \notag \\
&\hspace{1.2in} \displaystyle \in \mathcal{A}_{\epsilon}(X_R, Y_1, U_0, V_0, U_2, V_2).
\end{align}}
From \eqref{eq:TWRC:ULDL:typical:BS:index} and \eqref{eq:TWRC:ULDL:typical:UE1:index}, if the following rate constraints
\begin{align}   \label{eq:R3:CompressedMessage}
R_{3} &\leq \min \Big\{\mathrm{I}\left(X_R ; Y_0 | U_0, U_2, V_0, V_2, X_0 \right), \notag \\
      & \hspace{1.2in} \mathrm{I}\left(X_R ; Y_1 | U_0, U_2, V_0, V_2 \right) \Big\}
\end{align}	
and $n \rightarrow \infty $ are satisfied, then the decoding error probability can be made small \cite{book:Cover}. Then, each receiver calculates a set $L_{k}(\textbf{y}_k (b))$, for $k = 0, 2$, of $z$ such that $z \in L_{k}(\textbf{y}_{k,b})$ if
\begin{align}
&\{\hat{\textbf{y}}_R(z_{b}, w'_{0,c,b}, w'_{2,c,b}, w_{0,c,b}, w_{2,c,b}, s_{b}), \textbf{y}_{k, b},  \notag \\
&\hspace{0.2in}\textbf{v}_{0}(w'_{0,c, b}, w_{0,c,b}), \textbf{v}_{2}(w'_{2,c,b}, w_{2,c,b})), \textbf{x}_R(w'_{0,c,b}, w'_{2,c,b}, s_{b}) \notag \\
&\hspace{1.2in} \displaystyle \in \mathcal{A}_{\epsilon}(U_0, U_2, V_0, V_2, X_R, \hat{Y}_R, Y_{k}). \notag
\end{align}
Both the $\textrm{BS}$ and $\textrm{UE}_1$ declare that $\hat{z}_b$ was sent in block $b$ if
\begin{align}
\hat{z}_b \in \mathcal{S}_{3, s_{b+1}} \cap L_{0}(\textbf{y}_{0,b}) \notag \\
\hat{z}_b \in \mathcal{S}_{3, s_{b+1}} \cap L_{2}(\textbf{y}_{2,b}).
\end{align}
With arbitrarily small probability of error, we could have $\hat{z}_b = z_b$ if $n \rightarrow \infty $ and
\begin{subequations} \label{eq:ratesAtUE1_a}
\begin{align}
\hat{R}_3 + \epsilon &\leq \mathrm{I}\left(\hat{Y}_R ; Y_0 | V_0, V_2, X_R, X_0 \right) + R_3 - \epsilon, \\
\hat{R}_3 + \epsilon &\leq \mathrm{I}\left(\hat{Y}_R ; Y_1 | V_0, V_2, X_R \right) + R_3 - \epsilon.
\end{align}
\end{subequations}

From \eqref{eq:R3hat}, if we select $\hat{R}_3 = \mathrm{I}\left(\hat{Y}_R ; Y_R | X_R, V_0, V_2 \right) + \epsilon$, then \eqref{eq:ratesAtUE1_a} can be expressed as
{\footnotesize
\begin{align}   \label{eq:ratesAtUE1_c}
&\max \left\{\mathrm{I}\left(\hat{Y}_R ; Y_R | Y_0, X_0, X_R, V_0, V_2 \right), \mathrm{I} \left(\hat{Y}_R ; Y_R | Y_1, X_R, V_0, V_2\right) \right\} \notag \\
&\leq  \min \left\{\mathrm{I}\left(X_R ; Y_0 | U_0, U_2, V_0, V_2, X_0 \right), \mathrm{I}\left(X_R ; Y_1 | U_0, U_2, V_0, V_2 \right) \right\}.
\end{align}}	

Finally, the $\textrm{BS}$ uses both $\textbf{y}_{0}(b)$ and $\hat{\textbf{y}}_{R}(z_{b}, w'_{0,c,b}, w'_{2,c,b}, w_{0,c,b}, w_{2,c,b}, s_{b})$ to find an index $w_{2,d,b}$ such that
{\footnotesize
\begin{align*}
&\Big\{\textbf{u}_{0}(w'_{0,c,b}), \textbf{v}_{0}(w'_{0,c,b}, w_{0,c,b}), \textbf{x}_{0}(w'_{0,c,b}, w_{0,c,b}, w_{0,d,b}),  \notag \\
&\hspace{0.2in} \textbf{u}_{2}(w'_{2,c,b}), \textbf{v}_{2}(w'_{2,c,b}, w_{2,c,b}), \textbf{x}_{2}(w_{2,d,b}), \textbf{x}_{R}(w'_{0,c,b}, w'_{2,c,b}, s_{b}), \notag \\
&\hspace{0.3in} \hat{\textbf{y}}_{R}(z_{b}, w'_{0,c,b}, w'_{2,c,b}, w_{0,c,b}, w_{2,c,b}, s_{b}), \textbf{y}_{0}(b)\Big\}  \notag \\
&\hspace{1.1in} \in \mathcal{A}_{\epsilon}(U_0, U_2, V_0, V_2, X_0, X_2, X_R, \hat{Y}_R, Y_{0}).
\end{align*}}

Similarly, the $\textrm{UE}_1$ uses both $\textbf{y}_{1}(b)$ and $\hat{\textbf{y}}_{R}(z_{b}, w'_{0,c,b}, w'_{2,c,b}, w_{0,c,b}, w_{2,c,b}, s_{b})$ to find an index $w_{0,d,b}$ regarding $\epsilon$-typicality. 

Both receivers succeed with high probability if
\begin{subequations} \label{eq:rates4DirectTx:pDF+CF}
\begin{align}
R_{0,d} &\leq I(X_0; Y_1, \hat{Y}_R \;|\; U_0, U_2, V_0, V_2, X_R)       \\
R_{2,d} &\leq I(X_2; Y_0, \hat{Y}_R \;|\; U_0, U_2, V_0, V_2, X_0, X_R)
\end{align}
\end{subequations}
and $n \rightarrow \infty$. By combining \eqref{eq:ratesAtRN:pDF+CF}, \eqref{eq:ratesAtBS:pDF+CF}, \eqref{eq:ratesAtUE1:pDF+CF} and \eqref{eq:rates4DirectTx:pDF+CF} we get \eqref{eq:ITWRC:ULDL:pDF+CF}. This concludes the proof.

\end{appendices}

\section*{Acknowledgment}
This work was funded in part by the Mitsubishi Electric R\&D Center Europe (MERCE) and in part by the European Commission's 7th framework programme under grant agreement FP7-248993 also referred to as LOLA. 

\bibliographystyle{IEEEtran}
\bibliography{../erhan_biblio}

\end{document}